\documentclass[lettersize,journal,onecolumn]{IEEEtran}
\usepackage{PerlazaTIT2023}
\begin{document}
 
\title{Corrections to ``The Method of Gaps: \\Exact Expressions for the Generalization Error of Supervised Learning Algorithms''}
 
\author{Samir M. Perlaza{\orcidlink{0000-0002-1887-9215}} 
\thanks{
Samir M. Perlaza is with the College of Computing, Universit\'e Mohamed VI Polytechnique (UM6P), Ben Guerir, 43150, Morocco; also with the Centre Inria d'Universit\'e C\^ote d'Azur, INRIA, 06902 Sophia Antipolis, France; also with the Department of Electrical and Computer Engineering at  Princeton University, Princeton, NJ 08544 USA; and also with the GAATI Mathematics Laboratory at the University of French Polynesia, 98702 Faaa, French Polynesia (e-mail: samir.perlaza@um6p.ma). 
} 
}


 \maketitle 

\begin{abstract} 
In this paper, a counterexample to Lemma~5 of the paper referenced in the title is provided.
Consequently, the formulation of the method of data-driven gaps presented in the referenced paper relies on an additional assumption, namely, that the conditional distribution of the training dataset given the learned model has identical data-point-specific marginals.
To overcome this limitation, a generalized formulation of the method of data-driven gaps that does not rely on Lemma~5 and does not require the training samples to be independent or identically distributed is presented.
Through this formulation, the applicability of the method of data-driven gaps is extended to a broader class of data distributions and randomized learning algorithms.
\end{abstract}

\begin{IEEEkeywords}
Generalization Error, Supervised Learning, Worst-Case Data Generating Probability Measure; Gibbs Algorithm, Hypothesis tests, and Pythagorean Identities.
\end{IEEEkeywords}

\section{Context and Counter Example to Lemma~5 in~\cite{perlaza2024generalization}}

In~\cite[Lemma~5]{perlaza2024generalization}, it is asserted that, for an arbitrary randomized learning algorithm, independent and identically distributed (i.i.d.) training samples induce a conditional distribution of the training dataset given the learned model whose data-point-specific marginals are identical.
Although this property is satisfied by certain learning algorithms, particularly those that are invariant under permutations of the training data points, it does not hold in general.
For instance, Gibbs algorithms based on the empirical risk, as studied in~\cite{PerlazaTIT2024,aminian2024information,xu2017information} and references therein, satisfy this property because their learning rules are invariant under permutations of the training data points.
In contrast, Gibbs algorithms based on a weighted empirical risk, such as those studied in~\cite{bermudez2026decentralized,bermudez2026Equivalence,bermudez2026machine}, need not satisfy this property when different weights are assigned to different training data points.

The formulation of the method of data-driven gaps presented in~\cite{perlaza2024generalization}, which constitutes one of the two variants of the method of gaps, relies on~\cite[Lemma~5]{perlaza2024generalization}.
Consequently, the results derived from this lemma require an additional assumption, namely, that the conditional distribution of the training dataset given the learned model has identical data-point-specific marginals.
In particular, the equality asserted in~\cite[Equation~39]{perlaza2024generalization} must be imposed as an additional assumption for the results in~\cite[Lemmas~6 and~7]{perlaza2024generalization} and~\cite[Theorems~16--18, 20, and~21]{perlaza2024generalization}, insofar as their derivations rely on~\cite[Lemma~5]{perlaza2024generalization}.

To overcome this limitation, a generalized formulation of the method of data-driven gaps that does not rely on~\cite[Lemma~5]{perlaza2024generalization} is presented.
Specifically, Lemmas~\ref{LemmaOctober13at23h53in2026Benguerir} and~\ref{LemmaOctober14at16h00in2026Benguerir} generalize~\cite[Lemmas~6 and~7]{perlaza2024generalization}, respectively, without requiring the training data points to be independent or identically distributed, or their conditional marginal distributions given the learned model to be identical.
These generalizations extend the applicability of the method of data-driven gaps to a broader class of training-data distributions and randomized learning algorithms.

Building on Lemmas~\ref{LemmaOctober13at23h53in2026Benguerir} and~\ref{LemmaOctober14at16h00in2026Benguerir}, generalized versions of~\cite[Theorems~16--18, 20, and~21]{perlaza2024generalization} are also established.
These results retain the corresponding expressions for the generalization error while replacing the data-point-specific marginal and conditional marginal measures with their respective averages.
As a consequence, neither the i.i.d. assumption on the training data nor the equality of their conditional marginal distributions given the learned model is required.

\subsection{The Error in the Proof of~\cite[Lemma~5]{perlaza2024generalization}}

The proof of~\cite[Lemma~5]{perlaza2024generalization} correctly establishes in~\cite[Equation~44]{perlaza2024generalization} that, if $P_{\vect{Z}} \in \simplex{\left(\set{X}\times\set{Y}\right)^n}$ is the product measure induced by a probability measure $P_Z \in \simplex{\set{X}\times\set{Y}}$, as assumed in~\cite[Equation~22]{perlaza2024generalization}, then, for all $m \inCountK{n}$ and all measurable sets $\set{A}\subseteq\set{X}\times\set{Y}$,
\begin{IEEEeqnarray}{rcl}
P_Z\left(\set{A}\right)
& = &
\int P_{Z_m|\vect{\Theta}=\vect{\theta}}\left(\set{A}\right)
\mathrm{d}P_{\vect{\Theta}}\left(\vect{\theta}\right),
\end{IEEEeqnarray}
where $P_{\vect{\Theta}}$ is defined in~\cite[Equation~28]{perlaza2024generalization}, and $P_{Z_m|\vect{\Theta}}$, defined in~\cite[Equation~38]{perlaza2024generalization}, is a conditional marginal of the probability measure $P_{\vect{Z}|\vect{\Theta}}$ in~\cite[Equation~35]{perlaza2024generalization}.

The error in the proof arises from an incorrect application of~\cite[Theorem~5.3.1]{ash2000probability}. Specifically, it is asserted that, for any pair $(i,j)\inCountK{n}^2$, the equality
\begin{IEEEeqnarray}{rcl}
\int P_{Z_i|\vect{\Theta}=\vect{\theta}}\left(\set{A}\right)
\mathrm{d}P_{\vect{\Theta}}\left(\vect{\theta}\right)
& = &
\int P_{Z_j|\vect{\Theta}=\vect{\theta}}\left(\set{A}\right)
\mathrm{d}P_{\vect{\Theta}}\left(\vect{\theta}\right),
\end{IEEEeqnarray}
which holds for every measurable set $\set{A}\subseteq\set{X}\times\set{Y}$, implies that
\begin{IEEEeqnarray}{rcl}
\label{EqOctober17at13h22in2026Benguerir}
P_{\vect{\Theta}}\left(
\left\{
\vect{\theta}\in\set{M}:
P_{Z_i|\vect{\Theta}=\vect{\theta}}\left(\set{A}\right)
=
P_{Z_j|\vect{\Theta}=\vect{\theta}}\left(\set{A}\right)
\right\}
\right)
& = & 1.
\end{IEEEeqnarray}
However, equality of the integrals over the entire set of models $\set{M}$ does not imply almost-sure equality of their integrands. Consequently, the conclusion in~\eqref{EqOctober17at13h22in2026Benguerir} is not justified, as further demonstrated by the counterexample presented below.

For a fixed measurable set $\set{A}\subseteq\set{X}\times\set{Y}$, from~\cite[Theorem~1.6.11]{ash2000probability}, it is sufficient for the conclusion in~\eqref{EqOctober17at13h22in2026Benguerir} to hold that, for every measurable set $\set{C}\subseteq\set{M}$,
\begin{IEEEeqnarray}{rcl}
\label{EqOctober17at14h35in2026Benguerir}
\int_{\set{C}} P_{Z_i|\vect{\Theta}=\vect{\theta}}\left(\set{A}\right)
\mathrm{d}P_{\vect{\Theta}}\left(\vect{\theta}\right)
& = &
\int_{\set{C}} P_{Z_j|\vect{\Theta}=\vect{\theta}}\left(\set{A}\right)
\mathrm{d}P_{\vect{\Theta}}\left(\vect{\theta}\right),
\end{IEEEeqnarray}
which is a condition that is not assumed in the statement  of~\cite[Lemma~5]{perlaza2024generalization}.  The assumptions of~\cite[Lemma~5]{perlaza2024generalization} guarantee the equality in~\eqref{EqOctober17at14h35in2026Benguerir} only when $\set{C}=\set{M}$, which is insufficient to establish the conclusion in~\eqref{EqOctober17at13h22in2026Benguerir}.

\subsection{The Counter-Example}
 
 Consider the case in which $\set{X} = \set{Y} = \lbrace 0,1 \rbrace$ and $\set{M} = \lbrace 0, \nu_1, \nu_2, \nu_1 + \nu_2 \rbrace$, with $0 < \nu_1 < \nu_2 < +\infty$. 
%
%Consider a dataset formed by two datapoints $\vect{z} = \left( (0,0) , (1,1) \right) \in \left( \set{X} \times \set{Y} \right)^2$.
%
Consider also an algorithm $P_{\vect{\Theta} | \vect{Z}}$ such that for all $\vect{\theta} \in \set{M}$,
\begin{IEEEeqnarray}{rlc}
P_{\vect{\Theta} | \vect{Z} = \vect{z}} \left( \left\lbrace \vect{\theta} \right\rbrace \right) = \ind{\vect{\theta} = x_1y_1 \nu_1 + x_2y_2 \nu_2},
\end{IEEEeqnarray}
where $\vect{z} = \left( (x_1,y_1) , (x_2,y_2) \right)$. 
Datasets are sampled from a product measure $P_{\vect{Z}} \in \simplex{\left( \set{X} \times \set{Y} \right)^2}$ induced by a measure $P_{Z} \in \simplex{\set{X} \times \set{Y}}$ that satisfies
\begin{IEEEeqnarray}{rlc}
\label{EqOctober17at17h50in2026atBenGuerirA}
P_{Z} \left( \left\lbrace (0,0) \right\rbrace \right) = 1 - P_{Z} \left( \left\lbrace (1,1) \right\rbrace \right) = \beta,
\end{IEEEeqnarray}
for some $\beta \in (0,1)$.
Hence, the measure~$P_{\vect{\Theta}}$ in~\cite[Equation~28]{perlaza2024generalization} satisfies for all $\vect{\theta} \in \set{M}$,
\begin{IEEEeqnarray}{rCl}
\label{EqOctober17at17h50in2026atBenGuerirB}
P_{\vect{\Theta}} \left( \lbrace \vect{\theta} \rbrace \right) & = &  \sum_{\left( z_1, z_2 \right) \in \left( \set{X}\times\set{Y} \right)^2} P_{\vect{\Theta} | \vect{Z} = \left( z_1, z_2 \right)} \left( \left\lbrace \vect{\theta} \right\rbrace \right)  P_{Z} \left( \left\lbrace  z_1 \right\rbrace \right)P_{Z} \left( \left\lbrace  z_2 \right\rbrace \right) \\
\label{EqOctober17at17h50in2026atBenGuerirC}
& = & \left\lbrace
\begin{array}{ccc}
\beta (1- \beta) & \text{ if } & \vect{\theta} = \nu_1 \\
\beta (1- \beta) & \text{ if } & \vect{\theta} = \nu_2 \\
(1- \beta)^{2} & \text{ if } & \vect{\theta} = \nu_1 + \nu_2 \\
\beta^{2} & \text{ if } & \vect{\theta} = 0.
\end{array}
\right.
\end{IEEEeqnarray}
Moreover, for all $\vect{\theta} \in \set{M}$ and for all $\hat{z} \in \set{X} \times \set{Y}$,
\begin{IEEEeqnarray}{rCl}
P_{\vect{\Theta}|Z_1 = \hat{z}} \left( \lbrace \vect{\theta} \rbrace \right) & = &  \sum_{ z_2  \in \set{X}\times\set{Y} } P_{\vect{\Theta} | \vect{Z} = \left( \hat{z}, z_2 \right)} \left( \left\lbrace \vect{\theta} \right\rbrace \right) P_{Z} \left( \left\lbrace  z_2 \right\rbrace \right) \mbox{ and }\\
P_{\vect{\Theta}|Z_2 = \hat{z}} \left( \lbrace \vect{\theta} \rbrace \right) & = &  \sum_{ z_1  \in \set{X}\times\set{Y} } P_{\vect{\Theta} | \vect{Z} = \left( z_1, \hat{z}  \right)} \left( \left\lbrace \vect{\theta} \right\rbrace \right) P_{Z} \left( \left\lbrace  z_1 \right\rbrace \right).
\end{IEEEeqnarray}
Thus, using Bayes' rule, for all $i \in \lbrace 1,2 \rbrace$, 
\begin{IEEEeqnarray}{rCl}
\label{EqOctober17at18h15in2026atBenGuerirA}
P_{Z_i | \vect{\Theta} = \vect{\theta}} \left( \lbrace \hat{z} \rbrace \right) & = & \frac{P_{\vect{\Theta}|Z_i = \hat{z}} \left( \lbrace \vect{\theta} \rbrace \right) P_{Z} \left( \left\lbrace \hat{z} \right\rbrace \right) }{P_{\vect{\Theta}} \left( \lbrace \vect{\theta} \rbrace \right) }.
\end{IEEEeqnarray}
Hence, using \eqref{EqOctober17at18h15in2026atBenGuerirA}, it follows that  for all $i \in \lbrace 1,2 \rbrace$ and for all $\vect{\theta} \in \set{M}$,
\begin{IEEEeqnarray}{rCl}
P_{Z_i | \vect{\Theta} = \vect{\theta}} \left( \lbrace (0,1) \rbrace \right) = P_{Z_i | \vect{\Theta} = \vect{\theta}} \left( \lbrace (1,0) \rbrace \right) & = &0,
\end{IEEEeqnarray}
which follows from the fact that $P_{Z} \left( \lbrace (0,1) \rbrace \right) = P_{Z} \left( \lbrace (1,0) \rbrace \right) = 0$.

In the particular case of $\hat{z} = (0,0)$, it follows that
\begin{IEEEeqnarray}{rCl}
P_{\vect{\Theta}|Z_1 = (0,0)} \left( \lbrace \vect{\theta} \rbrace \right) & = &  \sum_{ z_2  \in \set{X}\times\set{Y} }  \ind{\vect{\theta} =   x_2y_2 \nu_2} P_{Z} \left( \left\lbrace  z_2 \right\rbrace \right) \\
\label{EqOctober17at17h50in2026atBenGuerirD}
& = & \left\lbrace
\begin{array}{ccc}
0 & \text{ if } & \vect{\theta} = \nu_1 \\
1-\beta & \text{ if } & \vect{\theta} = \nu_2 \\
0 & \text{ if } & \vect{\theta} = \nu_1 + \nu_2 \\
\beta & \text{ if } & \vect{\theta} = 0
\end{array}
\right. 
\end{IEEEeqnarray}
and
\begin{IEEEeqnarray}{rCl}
P_{\vect{\Theta}|Z_2 = (0,0)} \left( \lbrace \vect{\theta} \rbrace \right) & = & \sum_{ z_1  \in \set{X}\times\set{Y} }  \ind{\vect{\theta} =   x_1y_1 \nu_1} P_{Z} \left( \left\lbrace  z_1 \right\rbrace \right)  \\
\label{EqOctober17at17h50in2026atBenGuerirE}
& = & \left\lbrace
\begin{array}{ccc}
1-\beta & \text{ if } & \vect{\theta} = \nu_1 \\
0 & \text{ if } & \vect{\theta} = \nu_2 \\
0 & \text{ if } & \vect{\theta} = \nu_1 + \nu_2 \\
\beta & \text{ if } & \vect{\theta} = 0.
\end{array}
\right. 
\end{IEEEeqnarray}

Hence, using \eqref{EqOctober17at17h50in2026atBenGuerirA}, \eqref{EqOctober17at17h50in2026atBenGuerirC}, and~\eqref{EqOctober17at17h50in2026atBenGuerirD} in~\eqref{EqOctober17at18h15in2026atBenGuerirA} yields
\begin{IEEEeqnarray}{rCl}
\label{EqOctober17at19h06in2026atBenGuerirA}
P_{Z_1 | \vect{\Theta}  = \vect{\theta}} \left( \lbrace(0,0)\rbrace \right) 
& = & \left\lbrace
\begin{array}{ccc}
0 & \text{ if } & \vect{\theta} = \nu_1 \\
1 & \text{ if } & \vect{\theta} = \nu_2 \\
0 & \text{ if } & \vect{\theta} = \nu_1 + \nu_2 \\
1 & \text{ if } & \vect{\theta} = 0.
\end{array}
\right.
\end{IEEEeqnarray}
Similarly, using \eqref{EqOctober17at17h50in2026atBenGuerirA}, \eqref{EqOctober17at17h50in2026atBenGuerirC}, and~\eqref{EqOctober17at17h50in2026atBenGuerirE} in~\eqref{EqOctober17at18h15in2026atBenGuerirA} yields 
\begin{IEEEeqnarray}{rCl}
\label{EqOctober17at19h06in2026atBenGuerirB}
P_{Z_2 | \vect{\Theta}  = \vect{\theta}} \left( \lbrace(0,0)\rbrace \right) 
& = & \left\lbrace
\begin{array}{ccc}
1 & \text{ if } & \vect{\theta} = \nu_1 \\
0 & \text{ if } & \vect{\theta} = \nu_2 \\
0 & \text{ if } & \vect{\theta} = \nu_1 + \nu_2 \\
1 & \text{ if } & \vect{\theta} = 0.
\end{array}
\right.
\end{IEEEeqnarray}
Hence, from~\eqref{EqOctober17at19h06in2026atBenGuerirA} and~\eqref{EqOctober17at19h06in2026atBenGuerirB} , it follows that for all $\vect{\theta}\in\{\nu_1,\nu_2\}$,
\begin{IEEEeqnarray}{rCl}
P_{Z_1 | \vect{\Theta}  = \vect{\theta}} \left( \lbrace(0,0)\rbrace \right) 
& \neq & P_{Z_2 | \vect{\Theta}  = \vect{\theta}} \left( \lbrace(0,0)\rbrace \right).
\end{IEEEeqnarray}
Moreover, from~\eqref{EqOctober17at17h50in2026atBenGuerirC}, the set $\lbrace \nu_1,\nu_2 \rbrace$ has probability $2\beta(1-\beta)>0$ under $P_{\vect{\Theta}}$. Consequently, the data-point-specific conditional marginal measures are not identical $P_{\vect{\Theta}}$-almost surely, which contradicts the assertion of~\cite[Lemma~5]{perlaza2024generalization}.
  
\section{The Method of Data-Driven Gaps}

In this section, the two steps of the method of data-driven gaps are reformulated to accommodate a broader class of training-data distributions $P_{\vect{Z}}$ and randomized learning algorithms $P_{\vect{\Theta}|\vect{Z}}$.
This leads to more general expressions for the generalization error~$\overline{\overline{\mathsf{G}}} \left(P_{\vect{\Theta} | \vect{Z}}, P_{\vect{Z}} \right)$ in~\cite[Equation~11]{perlaza2024generalization} and avoids using~\cite[Lemma~5]{perlaza2024generalization}.

\subsection{Step One}
The first step of the \emph{method of data-driven gaps} consists in writing the generalization error~$\overline{\overline{\mathsf{G}}} \left(P_{\vect{\Theta} | \vect{Z}}, P_{\vect{Z}} \right)$ in~\cite[Equation~11]{perlaza2024generalization} in terms of the functional $\mathsf{G}: \set{M} \times \simplex{\set{X} \times \set{Y} } \times \simplex{\set{X} \times \set{Y}}$ in~\cite[Equation~19]{perlaza2024generalization}.
To achieve such an objective, consider the marginal measures of the probability measure~$P_{\vect{Z}}$, denoted by~$P_{Z_m} \in \triangle\left( \set{X} \times \set{Y} \right)$, with $m \inCountK{n}$, and note that for all measurable subsets $\set{A}\subseteq \left( \set{X} \times \set{Y} \right)$,
 \begin{IEEEeqnarray}{rcl}
\label{EqOctober13at22h54in2026Benguerir}
P_{Z_m} \left( \set{A} \right) & = & P_{\vect{Z}} \left( \set{B}_m \right),
\end{IEEEeqnarray}
where the set~$\set{B}_m \subseteq \left( \set{X} \times \set{Y} \right)^n$ is such that
\begin{IEEEeqnarray}{rcl}
\label{EqNovember25at09h37in2025Sophia}
\set{B}_m & = & \set{A}_{1} \times \set{A}_2 \times \ldots \times \set{A}_m \times \ldots \times \set{A}_n,
\end{IEEEeqnarray}
with~$\set{A}_t = \set{X} \times \set{Y}$ if~$t \in \lbrace 1, 2, \ldots, n \rbrace \setminus \lbrace m \rbrace$, and~$\set{A}_m = \set{A}$.

Consider also the marginal measures of the probability measure~$P_{\vect{Z} | \vect{\Theta} = \vect{\theta}}$ in~\cite[Equation~35]{perlaza2024generalization}, denoted by~$P_{Z_m | \vect{\Theta} = \vect{\theta}} \in \triangle\left( \set{X} \times \set{Y}\right)$, with $\vect{\theta} \in \set{M}$ and~$m \in \lbrace 1,2, \ldots, n \rbrace$, which satisfies for all measurable subsets $\set{A}\subseteq \left( \set{X} \times \set{Y} \right)$,
 \begin{IEEEeqnarray}{rcl}
P_{Z_m | \vect{\Theta} = \vect{\theta}} \left( \set{A} \right) & = & P_{\vect{Z} | \vect{\Theta} = \vect{\theta}} \left( \set{B}_m \right),
\end{IEEEeqnarray}
where the set~$\set{B}_m$ is defined in \eqref{EqNovember25at09h37in2025Sophia}.
Consider also the measures~$\bar{P}_{Z} \in \triangle\left( \set{X} \times \set{Y} \right)$ and~$\bar{P}_{Z_m | \vect{\Theta} = \vect{\theta}} \in \triangle\left( \set{X} \times \set{Y} | \set{M} \right)$ such that for all measurable sets $\set{A}$,
\begin{IEEEeqnarray}{rcl}
\label{EqOctober13at23h16in2026BenguerirA}
\bar{P}_{Z}\left( \set{A} \right) & = & \frac{1}{n} \sum_{m =1}^{n} P_{Z_m}\left( \set{A} \right)  \mbox{ and } \\ 
\label{EqOctober13at23h16in2026BenguerirB}
\bar{P}_{Z | \vect{\Theta} = \vect{\theta}} \left( \set{A} \right)  & = & \frac{1}{n} \sum_{m =1}^{n} P_{Z_m | \vect{\Theta} = \vect{\theta}}\left( \set{A} \right).
\end{IEEEeqnarray}
These measures $\bar{P}_{Z}$ in~\eqref{EqOctober13at23h16in2026BenguerirA} and $\bar{P}_{Z | \vect{\Theta} = \vect{\theta}}$ in~\eqref{EqOctober13at23h16in2026BenguerirB} are well-defined probability measures even when data points are neither independent nor identically distributed, and when the learning algorithm is not invariant under permutations of the training data points.
These measures can be respectively interpreted as the ``average'' marginal measure of the joint measure $P_{\vect{Z}}$ and the average conditional marginal measure given a model $\vect{\theta}$ of the joint conditional measure $P_{\vect{Z}| \vect{\Theta} = \vect{\theta}}$. 
Using this notation, the following lemma formalizes the first step of the method of data-driven gaps and shows that the generalization error~$\overline{\overline{\mathsf{G}}} \left(P_{\vect{\Theta} | \vect{Z}}, P_{\vect{Z}} \right)$ in~\cite[Equation~11]{perlaza2024generalization} depends exclusively on such ``average'' measures $\bar{P}_{Z}$ and $\bar{P}_{Z | \vect{\Theta}}$.
\begin{lemma}\label{LemmaOctober13at23h53in2026Benguerir}
Consider the generalization error~$\overline{\overline{\mathsf{G}}}(P_{\vect{\Theta} | \vect{Z}}, P_{\vect{Z}} )$ in~\cite[Equation~11]{perlaza2024generalization} and assume that 
\begin{itemize}
\item[$(a)$]  For all~$\vect{z} \in \left( \set{X} \times \set{Y} \right)^n$, the probability measures~$P_{\vect{\Theta} | \vect{Z} = \vect{z}}$ and~$P_{\vect{\Theta}}$ in~\cite[Equation~28]{perlaza2024generalization} are mutually absolutely continuous; and
\item[$(b)$]  For all~$\vect{\theta} \in \set{M}$, the probability measures~$P_{\vect{Z} | \vect{\Theta} = \vect{\theta}}$ and~$P_{\vect{Z}}$ in~\cite[Equation~35]{perlaza2024generalization} are mutually absolutely continuous.
\end{itemize}
Then,
	\begin{IEEEeqnarray}{rCl}
\label{EqOctober21at11h08in2024Agadir}	
		\overline{\overline{\mathsf{G}}}(P_{\vect{\Theta} | \vect{Z}},P_{\vect{Z}} )  & = &  \int \mathsf{G}\left( \vect{\theta}, \bar{P}_{Z}, \bar{P}_{Z | \vect{\Theta} = \vect{\theta}}\right) \mathrm{d} P_{\vect{\Theta}}\left( \vect{\theta}\right),
	\end{IEEEeqnarray}
	where the functional~$\mathsf{G}$ is defined in~\cite[Equation~19]{perlaza2024generalization}; and the probability measures $\bar{P}_{Z}$ and $\bar{P}_{Z | \vect{\Theta} = \vect{\theta}}$ are defined in~\eqref{EqOctober13at23h16in2026BenguerirA} and~\eqref{EqOctober13at23h16in2026BenguerirB}, respectively.
\end{lemma}
\begin{IEEEproof}
	The proof is presented in Appendix~\ref{AppProofOfLemmaOctober13at23h53in2026Benguerir}.	
\end{IEEEproof}

Lemma~\ref{LemmaOctober13at23h53in2026Benguerir} is a generalization of~\cite[Lemma~6]{perlaza2024generalization} that does not rely on~\cite[Lemma~5]{perlaza2024generalization}.

\subsection{Step Two}
The second step of the method of data-driven gaps is formalized as follows.

\begin{lemma}\label{LemmaOctober14at16h00in2026Benguerir}  
Consider the generalization error~$\overline{\overline{\mathsf{G}}} \left(P_{\vect{\Theta} | \vect{Z}}, P_{\vect{Z}} \right)$ in~\cite[Equation~11]{perlaza2024generalization} and assume that: 
\begin{itemize}
\item[$(a)$]  For all~$\vect{z} \in \left( \set{X} \times \set{Y} \right)^n$, the probability measures~$P_{\vect{\Theta} | \vect{Z} = \vect{z}}$ and~$P_{\vect{\Theta}}$ in~\cite[Equation~28]{perlaza2024generalization} are mutually absolutely continuous; 
\item[$(b)$]  For all~$\vect{\theta} \in \set{M}$, the probability measures~$P_{\vect{Z} | \vect{\Theta} = \vect{\theta}}$ and~$P_{\vect{Z}}$ in~\cite[Equation~35]{perlaza2024generalization} are mutually absolutely continuous; and
\item[$(c)$] For all~$\vect{\theta} \in \set{M}$,  the probability measures~$\bar{P}_{Z}$ in~\eqref{EqOctober13at23h16in2026BenguerirA} and $\bar{P}_{Z | \vect{\Theta} = \vect{\theta}}$ in~\eqref{EqOctober13at23h16in2026BenguerirB} are both mutually absolutely continuous with respect to~$P_S \in \triangle\left( \set{X} \times \set{Y} \right)$ in~\cite[Equation~51]{perlaza2024generalization}.
\end{itemize}
Then, 
\begin{IEEEeqnarray}{rCl}
		\nonumber
& &\overline{\overline{\mathsf{G}}}(P_{\vect{\Theta} | \vect{Z}}, P_{\vect{Z}} )  
		= \beta \int \Bigg( \KL{\bar{P}_{Z | \vect{\Theta} = \vect{\theta}}}{P_{\hat{Z} | \vect{\Theta} = \vect{\theta}}^{\left( P_{S}, \beta \right)} }-\KL{\bar{P}_{Z}}{P_{\hat{Z} | \vect{\Theta} = \vect{\theta}}^{\left( P_{S}, \beta \right)} }\Tsupersqueezeequ \\
\label{EqOctober19at18h32in2024InTheAirFlyingToAgadir}
		& & - \KL{\bar{P}_{Z | \vect{\Theta} = \vect{\theta}}}{P_S} + \KL{\bar{P}_Z}{P_S}  \Bigg) \mathrm{d} P_{\vect{\Theta}}\left( \vect{\theta} \right), \squeezeequ \IEEEeqnarraynumspace
	\end{IEEEeqnarray}
where the probability measure~$P_{\hat{Z} | \vect{\Theta} = \vect{\theta}}^{\left( P_{S}, \beta \right)}~$ is defined in~\cite[Equation~51]{perlaza2024generalization}; and~$\beta \in \set{J}_{P_S}$, with~$\set{J}_{P_S}$ in~\cite[Equation~50]{perlaza2024generalization}.
\end{lemma}
\begin{IEEEproof}
From Assumptions~$(a)$ and~$(b)$, Lemma~\ref{LemmaOctober13at23h53in2026Benguerir} holds. 
From Assumptions~$(c)$, it follows from~\cite[Theorem~32]{perlaza2024generalization}, that for all $\vect{\theta} \in \set{M}$,
\begin{IEEEeqnarray}{rCl}
		\nonumber
& &\mathsf{G}\left( \vect{\theta}, \bar{P}_{Z}, \bar{P}_{Z | \vect{\Theta} = \vect{\theta}}\right) 
		= \beta \Bigg( \KL{\bar{P}_{Z | \vect{\Theta} = \vect{\theta}}}{P_{\hat{Z} | \vect{\Theta} = \vect{\theta}}^{\left( P_{S}, \beta \right)} }-\KL{\bar{P}_{Z}}{P_{\hat{Z} | \vect{\Theta} = \vect{\theta}}^{\left( P_{S}, \beta \right)} }\Tsupersqueezeequ \\
\label{EqOctober16at17h29in2026InBenguerir}
		& & - \KL{\bar{P}_{Z | \vect{\Theta} = \vect{\theta}}}{P_S} + \KL{\bar{P}_Z}{P_S}  \Bigg), \squeezeequ \IEEEeqnarraynumspace
	\end{IEEEeqnarray}
where the functional~$\mathsf{G}$ is defined in~\cite[Equation~19]{perlaza2024generalization}.
Hence, using~\eqref{EqOctober16at17h29in2026InBenguerir} in~\eqref{EqOctober21at11h08in2024Agadir} completes the proof.
\end{IEEEproof}

Lemma~\ref{LemmaOctober14at16h00in2026Benguerir} is a generalization of~\cite[Lemma~7]{perlaza2024generalization} that does not rely on~\cite[Lemma~5]{perlaza2024generalization}.
  
\section{Explicit Expressions Obtained Via the Method of Data-Driven Gaps}\label{SecExpressionsViaDataDrivenGaps}

All explicit expressions for the generalization error $\overline{\overline{\mathsf{G}}}\left(P_{\vect{\Theta}|\vect{Z}},P_{\vect{Z}}\right)$ defined in~\cite[Equation~11]{perlaza2024generalization} and presented in Section~VI of~\cite{perlaza2024generalization} can be generalized using Lemmas~\ref{LemmaOctober13at23h53in2026Benguerir} and~\ref{LemmaOctober14at16h00in2026Benguerir}, without relying on~\cite[Lemma~5]{perlaza2024generalization}.

The key distinction between the generalized expressions and their original counterparts lies in the probability measures involved. Specifically, the generalized expressions do not explicitly involve data-point-specific marginal measures of $P_{\vect{Z}}$ or data-point-specific conditional marginal measures of $P_{\vect{Z}|\vect{\Theta}}$. Instead, they are formulated in terms of the average of all data-point-specific marginal measures, namely, the measure $\bar{P}_{Z}$ defined in~\eqref{EqOctober13at23h16in2026BenguerirA}, and the average of all data-point-specific conditional marginal measures, namely, the measure $\bar{P}_{Z|\vect{\Theta}}$ defined in~\eqref{EqOctober13at23h16in2026BenguerirB}. Consequently, the training data points are required neither to be independent nor to be identically distributed, and their conditional marginal distributions given the learned model are not required to be identical.

The proofs of these generalized expressions follow the same lines as those presented in~\cite{perlaza2024generalization}, with only minor algebraic modifications. Accordingly, only the steps that differ from the original proofs are detailed.
Furthermore, since the analysis of the various equivalent expressions for the generalization error $\overline{\overline{\mathsf{G}}}\left(P_{\vect{\Theta}|\vect{Z}},P_{\vect{Z}}\right)$ presented in~\cite{perlaza2024generalization} extends naturally to their generalized counterparts, no additional analysis is provided. Instead, the presentation is focused on clearly stating the generalized results and establishing their validity through concise proofs.

The following is a generalization of~\cite[Theorem~16]{perlaza2024generalization}.

\begin{theorem}\label{TheoremMismatchP_S}
Consider the generalization error~$\overline{\overline{\mathsf{G}}} \left(P_{\vect{\Theta} | \vect{Z}}, P_{\vect{Z}} \right)$ in~\cite[Equation~11]{perlaza2024generalization} and assume that: 
\begin{itemize}
\item[$(a)$]  For all~$\vect{z} \in \left( \set{X} \times \set{Y} \right)^n$, the probability measures~$P_{\vect{\Theta} | \vect{Z} = \vect{z}}$ and~$P_{\vect{\Theta}}$ in~\cite[Equation~28]{perlaza2024generalization} are mutually absolutely continuous; 
\item[$(b)$]  For all~$\vect{\theta} \in \set{M}$, the probability measures~$P_{\vect{Z} | \vect{\Theta} = \vect{\theta}}$ and~$P_{\vect{Z}}$ in~\cite[Equation~35]{perlaza2024generalization} are mutually absolutely continuous; and
\item[$(c)$] For all~$\vect{\theta} \in \set{M}$,  the probability measures~$\bar{P}_{Z}$ in~\eqref{EqOctober13at23h16in2026BenguerirA} and $\bar{P}_{Z | \vect{\Theta} = \vect{\theta}}$ in~\eqref{EqOctober13at23h16in2026BenguerirB} are both mutually absolutely continuous with respect to~$P_S \in \triangle\left( \set{X} \times \set{Y} \right)$ in~\cite[Equation~51]{perlaza2024generalization}.
\end{itemize}
Then, 
	\begin{IEEEeqnarray}{rCl}
		\nonumber
		&&\overline{\overline{\mathsf{G}}}(P_{\vect{\Theta} | \vect{Z}}, P_{\vect{Z}} )  \\
		\nonumber
		&=& - \beta \Bigg(\iint \log{\left(\frac{\mathrm{d}P_{\hat{Z} | \vect{\Theta} = \vect{\theta}}^{\left( P_{S}, \beta \right)} }{\mathrm{d}P_S}(z)\right)} \mathrm{d}\bar{P}_{Z|\vect{\Theta} = \vect{\theta}}(z)\mathrm{d}P_{\vect{\Theta}}\left(\vect{\theta}\right) \\
		\label{EqMismatchP_S}
		&&+  \iint \log{\left(\frac{\mathrm{d}P_{\hat{Z} | \vect{\Theta} = \vect{\theta}}^{\left( P_{S}, \beta \right)} }{\mathrm{d}P_S}(z)\right)} \mathrm{d}\bar{P}_{Z}(z)\mathrm{d}P_{\vect{\Theta}}\left(\vect{\theta}\right)\Bigg),
	\end{IEEEeqnarray}
where the probability measure~$P_{\hat{Z} | \vect{\Theta} = \vect{\theta}}^{\left( P_{S}, \beta \right)}~$ is defined in~\cite[Equation~51]{perlaza2024generalization}; and~$\beta \in \set{J}_{P_S}$, with~$\set{J}_{P_S}$ in~\cite[Equation~50]{perlaza2024generalization}.
\end{theorem}
\begin{IEEEproof}
From Assumptions~$(a)$, $(b)$ and~$(c)$, Lemma~\ref{LemmaOctober14at16h00in2026Benguerir} holds. Then, by following the same steps as in the proof of~\cite[Theorem~16]{perlaza2024generalization}, it follows that the differences~$\KL{\bar{P}_{Z|\vect{\Theta}=\vect{\theta}}}{P_{\hat{Z} | \vect{\Theta} = \vect{\theta}}^{\left( P_{S}, \beta \right)} } - \KL{\bar{P}_{Z|\vect{\Theta}=\vect{\theta}}}{P_S}$ and~$\KL{\bar{P}_Z}{P_S} -\KL{\bar{P}_Z}{P_{\hat{Z} | \vect{\Theta} = \vect{\theta}}^{\left( P_{S}, \beta \right)} }$ in~\eqref{EqOctober19at18h32in2024InTheAirFlyingToAgadir} respectively satisfy the following equalities:
	\begin{IEEEeqnarray}{rCl} 
		\nonumber
	&&	\KL{\bar{P}_{Z|\vect{\Theta}=\vect{\theta}}}{P_{\hat{Z} | \vect{\Theta} = \vect{\theta}}^{\left( P_{S}, \beta \right)} } - \KL{\bar{P}_{Z|\vect{\Theta}=\vect{\theta}}}{P_S}\\
\label{EqOctober16at12h28in2026BenguerirA}
		&=&- \int \log{\frac{\mathrm{d}P_{\hat{Z} | \vect{\Theta} = \vect{\theta}}^{\left( P_{S}, \beta \right)} }{\mathrm{d}P_S}\left(z\right)} \mathrm{d}\bar{P}_{Z|\vect{\Theta}=\vect{\theta}}\left(z\right),
	\end{IEEEeqnarray}
and
	\begin{IEEEeqnarray}{rCl}
		\nonumber
&&\KL{\bar{P}_Z}{P_S} -\KL{\bar{P}_Z}{P_{\hat{Z} | \vect{\Theta} = \vect{\theta}}^{\left( P_{S}, \beta \right)} }\\
\label{EqOctober16at12h28in2026BenguerirB}
&=&\int \log{\frac{\mathrm{d}P_{\hat{Z} | \vect{\Theta} = \vect{\theta}}^{\left( P_{S}, \beta \right)} }{\mathrm{d}P_S}\left(z\right)} \mathrm{d}\bar{P}_{Z}\left(z\right).
	\end{IEEEeqnarray}
Plugging both~\eqref{EqOctober16at12h28in2026BenguerirA} and~\eqref{EqOctober16at12h28in2026BenguerirB} into~\eqref{EqOctober19at18h32in2024InTheAirFlyingToAgadir} yields~\eqref{EqMismatchP_S}, which completes the proof.
\end{IEEEproof}

The following is a generalization of~\cite[Theorem~17]{perlaza2024generalization}.

\begin{theorem}\label{TheoDecember23at7h51in2025atSophia}
Consider the generalization error~$\overline{\overline{\mathsf{G}}} \left(P_{\vect{\Theta} | \vect{Z}}, P_{\vect{Z}} \right)$ in~\cite[Equation~11]{perlaza2024generalization} and assume that: 
\begin{itemize}
\item[$(a)$]  For all~$\vect{z} \in \left( \set{X} \times \set{Y} \right)^n$, the probability measures~$P_{\vect{\Theta} | \vect{Z} = \vect{z}}$ and~$P_{\vect{\Theta}}$ in~\cite[Equation~28]{perlaza2024generalization} are mutually absolutely continuous; 
\item[$(b)$]  For all~$\vect{\theta} \in \set{M}$, the probability measures~$P_{\vect{Z} | \vect{\Theta} = \vect{\theta}}$ and~$P_{\vect{Z}}$ in~\cite[Equation~35]{perlaza2024generalization} are mutually absolutely continuous; and
\item[$(c)$] For all~$\vect{\theta} \in \set{M}$,  the probability measures~$\bar{P}_{Z}$ in~\eqref{EqOctober13at23h16in2026BenguerirA} and $\bar{P}_{Z | \vect{\Theta} = \vect{\theta}}$ in~\eqref{EqOctober13at23h16in2026BenguerirB} are both mutually absolutely continuous with respect to~$P_S \in \triangle\left( \set{X} \times \set{Y} \right)$ in~\cite[Equation~51]{perlaza2024generalization}.
\end{itemize}
Then,  
\begin{IEEEeqnarray}{C}
\nonumber
- \beta \left(  \bar{\eta}  - \eta \right) \bar{P}_{Z | \vect{\Theta} } P_{\vect{\Theta}} \left( \set{C}_{\eta} \right)  - \beta \left( \eta - \underline{\eta} \right) \bar{P}_{Z}P_{\vect{\Theta}} \left( \set{C}_{\eta}^{\sfc} \right)\\ 
\nonumber
 \leqslant \overline{\overline{\mathsf{G}}} \left(P_{\vect{\Theta} | \vect{Z}}, P_{\vect{Z}} \right) \leqslant\\
\label{EqDecember23at7h53in2025atSophia}
\beta \left(  \eta -  \underline{\eta} \right) \bar{P}_{Z | \vect{\Theta} } P_{\vect{\Theta}} \left( \set{C}_{\eta}^{\sfc}\right)  + \beta \left( \bar{\eta}  -  \eta \right) \bar{P}_{Z}P_{\vect{\Theta}} \left( \set{C}_{\eta} \right), \spnum 
\end{IEEEeqnarray}
where $\eta$ satisfies
\begin{IEEEeqnarray}{rCl}
\underline{\eta} & \leqslant \eta \leqslant & \bar{\eta},
\end{IEEEeqnarray}
and $\bar{\eta}$ and $\underline{\eta}$ are respectively defined in~\cite[Equations~125 and~126]{perlaza2024generalization}; and the set $\set{C}_{\eta}$ is defined in~\cite[Equation~122]{perlaza2024generalization}. 
\end{theorem}
\begin{IEEEproof}
The proof follows the same steps as those in the proof of~\cite[Theorem~17]{perlaza2024generalization}.
\end{IEEEproof}

The following is a generalization of~\cite[Theorem~18]{perlaza2024generalization}.

\begin{theorem}\label{Corollarycomparisonnegative}
Consider the generalization error~$\overline{\overline{\mathsf{G}}} \left(P_{\vect{\Theta} | \vect{Z}}, P_{\vect{Z}} \right)$ in~\cite[Equation~11]{perlaza2024generalization} and assume that: 
\begin{itemize}
\item[$(a)$]  For all~$\vect{z} \in \left( \set{X} \times \set{Y} \right)^n$, the probability measures~$P_{\vect{\Theta} | \vect{Z} = \vect{z}}$ and~$P_{\vect{\Theta}}$ in~\cite[Equation~28]{perlaza2024generalization} are mutually absolutely continuous; 
\item[$(b)$]  For all~$\vect{\theta} \in \set{M}$, the probability measures~$P_{\vect{Z} | \vect{\Theta} = \vect{\theta}}$ and~$P_{\vect{Z}}$ in~\cite[Equation~35]{perlaza2024generalization} are mutually absolutely continuous; and
\item[$(c)$] For all~$\vect{\theta} \in \set{M}$,  the probability measures~$\bar{P}_{Z}$ in~\eqref{EqOctober13at23h16in2026BenguerirA} and $\bar{P}_{Z | \vect{\Theta} = \vect{\theta}}$ in~\eqref{EqOctober13at23h16in2026BenguerirB} are both mutually absolutely continuous with respect to~$P_S \in \triangle\left( \set{X} \times \set{Y} \right)$ in~\cite[Equation~51]{perlaza2024generalization}.
\end{itemize}
Then,
	\begin{IEEEeqnarray}{rCl}
		\nonumber
		&&\overline{\overline{\mathsf{G}}}(P_{\vect{\Theta} | \vect{Z}},P_{\vect{Z}} )= - \beta \left(  I\left(\bar{P}_{Z | \vect{\Theta}}; P_{\vect{\Theta}} \right)  +  L\left(\bar{P}_{Z | \vect{\Theta}}; P_{\vect{\Theta}} \right) \right)\squeezeequ\\
		\nonumber
		&&- \beta \iint \log{\left(\frac{\mathrm{d}P_{\hat{Z} | \vect{\Theta} = \vect{\theta}}^{\left( P_{S}, \beta \right)} }{\mathrm{d}\bar{P}_{Z | \vect{\Theta} = \vect{\theta}}}(z)\right)} \mathrm{d}\bar{P}_{Z | \vect{\Theta} = \vect{\theta}}\left(z\right)\mathrm{d}P_{\vect{\Theta}}\left(\vect{\theta}\right)\\
		\label{EqNovember17at11h16in2024Nice}
		&&+\beta \iint \log{\left(\frac{\mathrm{d}P_{\hat{Z} | \vect{\Theta} = \vect{\theta}}^{\left( P_{S}, \beta \right)} }{\mathrm{d}\bar{P}_{Z | \vect{\Theta} = \vect{\theta}}}(z)\right)} \mathrm{d}\bar{P}_{Z}\left(z\right)\mathrm{d}P_{\vect{\Theta}}\left(\vect{\theta}\right),\spnum
	\end{IEEEeqnarray}
where the probability measure~$P_{\hat{Z} | \vect{\Theta} = \vect{\theta}}^{\left( P_{S}, \beta \right)}~$ is defined in~\cite[Equation~51]{perlaza2024generalization}; and~$\beta \in \set{J}_{P_S}$, with~$\set{J}_{P_S}$ in~\cite[Equation~50]{perlaza2024generalization}.
\end{theorem}
	\begin{IEEEproof}
From Assumptions~$(a)$, $(b)$ and~$(c)$, Lemma~\ref{LemmaOctober14at16h00in2026Benguerir} holds. %
Then, by following the same steps as in the proof of~\cite[Theorem~18]{perlaza2024generalization}, it follows that the difference~$\int\left(\KL{\bar{P}_{Z | \vect{\Theta} = \vect{\theta}}}{P_S} - \KL{\bar{P}_Z}{P_S}  \right) \mathrm{d} P_{\vect{\Theta}}\left( \vect{\theta} \right)$ and the term $\int \KL{\bar{P}_{Z}}{P_{\hat{Z} | \vect{\Theta} = \vect{\theta}}^{\left( P_{S}, \beta \right)} } \mathrm{d} P_{\vect{\Theta}} \left( \vect{\theta} \right)$ 
in~\eqref{EqOctober19at18h32in2024InTheAirFlyingToAgadir} respectively satisfy the following equalities:
\begin{IEEEeqnarray}{rCl}
\nonumber
& &\int\left(\KL{\bar{P}_{Z | \vect{\Theta} = \vect{\theta}}}{P_S} - \KL{\bar{P}_Z}{P_S}  \right) \mathrm{d} P_{\vect{\Theta}}\left( \vect{\theta} \right) \\
\label{EqNovember13at07h06in2024InTheBusToSophia}
& = & I\left(\bar{P}_{Z | \vect{\Theta}}; P_{\vect{\Theta}} \right).    \spnum
\end{IEEEeqnarray}
and 
	\begin{IEEEeqnarray}{rCl}
		\nonumber
		&& \int \KL{\bar{P}_{Z}}{P_{\hat{Z} | \vect{\Theta} = \vect{\theta}}^{\left( P_{S}, \beta \right)} } \mathrm{d} P_{\vect{\Theta}} \left( \vect{\theta} \right) \\
&=& L\left( \bar{P}_{Z| \vect{\Theta}}; P_{\vect{\Theta}} \right)\\
		\label{EqOctober21at12h40in2024Agadir}
		&&+ \iint\log{\left(\frac{\mathrm{d}\bar{P}_{Z|\vect{\Theta} = \vect{\theta}}}{\mathrm{d}P_{\hat{Z} | \vect{\Theta} = \vect{\theta}}^{\left( P_{S}, \beta \right)} }(z)\right)} \mathrm{d}\bar{P}_{Z}\left(z\right)\mathrm{d}P_{\vect{\Theta}}\left(\vect{\theta}\right).
	\end{IEEEeqnarray}
The proof ends by using the equalities~\eqref{EqNovember13at07h06in2024InTheBusToSophia} and~\eqref{EqOctober21at12h40in2024Agadir} in~\eqref{EqOctober19at18h32in2024InTheAirFlyingToAgadir}, which yields~\eqref{EqNovember17at11h16in2024Nice}.
\end{IEEEproof}
  
The following theorem is a generalization of~\cite[Theorem~20]{perlaza2024generalization}.
  
\begin{theorem}\label{TheoremOctober16at19h15in2026atBenguerir}
Consider the generalization error~$\overline{\overline{\mathsf{G}}} \left(P_{\vect{\Theta} | \vect{Z}}, P_{\vect{Z}} \right)$ in~\cite[Equation~11]{perlaza2024generalization} and assume that:
\begin{itemize}
\item[$(a)$]  For all~$\vect{z} \in \left( \set{X} \times \set{Y} \right)^n$, the probability measures~$P_{\vect{\Theta} | \vect{Z} = \vect{z}}$ and~$P_{\vect{\Theta}}$ in~\cite[Equation~28]{perlaza2024generalization} are mutually absolutely continuous; 
\item[$(b)$]  For all~$\vect{\theta} \in \set{M}$, the probability measures~$P_{\vect{Z} | \vect{\Theta} = \vect{\theta}}$ and~$P_{\vect{Z}}$ in~\cite[Equation~35]{perlaza2024generalization} are mutually absolutely continuous; and
\item[$(c)$] For all~$\vect{\theta} \in \set{M}$,  the probability measures~$\bar{P}_{Z}$ in~\eqref{EqOctober13at23h16in2026BenguerirA} and $\bar{P}_{Z | \vect{\Theta} = \vect{\theta}}$ in~\eqref{EqOctober13at23h16in2026BenguerirB} are both mutually absolutely continuous with respect to~$P_S \in \triangle\left( \set{X} \times \set{Y} \right)$ in~\cite[Equation~51]{perlaza2024generalization}.
\end{itemize}
	Then,
\begin{IEEEeqnarray}{rcl}
	\nonumber
\overline{\overline{\mathsf{G}}} \left(P_{\vect{\Theta} | \vect{Z}}, P_{\vect{Z}} \right) \Dsupersqueezeequ &=& \beta \iint \bigg(\KL{\bar{P}_{Z|\vect{\Theta}=\vect{\mu}}}{P^{(P_{S},\beta)}_{\hat{Z}|\vect{\Theta}=\vect{\mu}}} \Dsupersqueezeequ\IEEEeqnarraynumspace\\
&&- \KL{\bar{P}_{Z|\vect{\Theta}=\vect{\mu}}}{P^{(P_{S},\beta)}_{\hat{Z}|\vect{\Theta}=\vect{\nu}} }\bigg)\mathrm{d}P_{\vect{\Theta}}(\vect{\nu})\mathrm{d}P_{\vect{\Theta}}(\vect{\mu}), \Dsupersqueezeequ\IEEEeqnarraynumspace
\end{IEEEeqnarray}
where the probability measure~$P_{\hat{Z} | \vect{\Theta} = \vect{\theta}}^{\left( P_{S}, \beta \right)}~$ is defined in~\cite[Equation~51]{perlaza2024generalization}; and~$\beta \in \set{J}_{P_S}$, with~$\set{J}_{P_S}$ in~\cite[Equation~50]{perlaza2024generalization}.
\end{theorem}
\begin{IEEEproof}
The proof of Theorem~\ref{TheoremOctober16at19h15in2026atBenguerir} is presented in Appendix~\ref{AppendixProofOfTheoremOctober16at19h15in2026atBenguerir}.
\end{IEEEproof}
 
The following lemma is a generalization of~\cite[Lemma~21]{perlaza2024generalization}.

\begin{lemma}\label{LemmaGeometricEntropy}
Consider two probability measures~$P_{\vect{Z}}\in \triangle\left( \left( \set{X} \times \set{Y}\right)^n\right)$ and~$P_{\vect{\Theta}}\in \triangle\left( \set{M} \right)$; and  two  conditional probability measures~$P_{\vect{\Theta}|\vect{Z}} \in \triangle\left( \set{M} | \left( \set{X} \times \set{Y}\right)^n\right)$ and~$P_{\vect{Z} | \vect{\Theta}} \in \triangle\left( \left( \set{X} \times \set{Y}\right)^n | \set{M} \right)$ and assume that:
\begin{itemize}
\item[$(a)$] The measures $P_{\vect{Z}}$, $P_{\vect{\Theta}}$, $P_{\vect{\Theta}|\vect{Z}}$ and $P_{\vect{Z} | \vect{\Theta}}$ satisfy~\cite[Equation~28]{perlaza2024generalization} and ~\cite[Equation~35]{perlaza2024generalization};  
\item[$(b)$]  For all~$\vect{\theta} \in \set{M}$, the probability measure~$P_{\vect{Z} | \vect{\Theta} = \vect{\theta}}$ is absolutely continuous with respect to~$P_{\vect{Z}}$; and 
\item[$(c)$]  For all~$\vect{\theta} \in \set{M}$,  the probability measures~$\bar{P}_{Z}$ in~\eqref{EqOctober13at23h16in2026BenguerirA} and $\bar{P}_{Z | \vect{\Theta} = \vect{\theta}}$ in~\eqref{EqOctober13at23h16in2026BenguerirB} are both mutually absolutely continuous with respect to~$P_S \in \triangle\left( \set{X} \times \set{Y} \right)$ in~\cite[Equation~51]{perlaza2024generalization}.
\end{itemize}
Then,
	\begin{IEEEeqnarray}{rCl}
		\nonumber
&&	\iint \KL{\bar{P}_{Z|\vect{\Theta}=\vect{\theta}}}{P^{(P_{S},\beta)}_{\hat{Z}|\vect{\Theta}=\vect{\nu}} }\bigg)\mathrm{d}P_{\vect{\Theta}}(\vect{\nu})\mathrm{d}P_{\vect{\Theta}}(\vect{\theta})\\
\nonumber
	&=&\int \bigg(\KL{\bar{P}_{Z|\vect{\Theta}=\vect{\theta}}}{\bar{P}_Z}  + \KL{\bar{P}_Z}{P^{(P_{S},\beta)}_{\hat{Z}|\vect{\Theta}=\vect{\theta}}} \bigg)\mathrm{d}P_{\vect{\Theta}}\left(\vect{\theta}\right),
	\end{IEEEeqnarray}
where the probability measure~$P_{\hat{Z} | \vect{\Theta} = \vect{\theta}}^{\left( P_{S}, \beta \right)}~$ is defined in~\cite[Equation~51]{perlaza2024generalization}; and~$\beta \in \set{J}_{P_S}$, with~$\set{J}_{P_S}$ in~\cite[Equation~50]{perlaza2024generalization}.
\end{lemma}
 \begin{IEEEproof}
Note that Assumption~$(c)$ implies that for all~$\vect{\theta} \in \set{M}$,  the probability measures~$\bar{P}_{Z}$ in~\eqref{EqOctober13at23h16in2026BenguerirA} and $\bar{P}_{Z | \vect{\Theta} = \vect{\theta}}$ in~\eqref{EqOctober13at23h16in2026BenguerirB} are mutually absolutely continuous. Using this observation, the proof follows  the same steps of the proof of~\cite[Lemma~21]{perlaza2024generalization} by replacing $P_{Z}$ and $P_{Z | \vect{\Theta}}$ by $\bar{P}_{Z}$ and $\bar{P}_{Z | \vect{\Theta}}$, respectively.
  \end{IEEEproof}
\bibliographystyle{IEEEtran}
\bibliography{references}
 
\appendices

\section{Proof of Lemma~\ref{LemmaOctober13at23h53in2026Benguerir}}\label{AppProofOfLemmaOctober13at23h53in2026Benguerir}

The proof is essentially algebraic and follows from~\cite[Definition~4]{perlaza2024generalization}. More specifically, 
\begin{IEEEeqnarray}{rcl}
\nonumber
& & \overline{\overline{\mathsf{G}}} \left(P_{\vect{\Theta} | \vect{Z}}, P_{\vect{Z}} \right) \\
\label{EqNovember8at14h22in2024SophiaAntipolisA}
& = & \iint  \left(  
\mathsf{R}_{\vect{u}} \left( P_{\vect{\Theta} | \vect{Z} = \vect{z}} \right) - 
\mathsf{R}_{\vect{z}}\left( P_{\vect{\Theta} | \vect{Z} = \vect{z}} \right)
 \right) \mathrm{d}P_{\vect{Z}} \left( \vect{u} \right) \mathrm{d}P_{\vect{Z}} \left( \vect{z} \right) \squeezeequ \IEEEeqnarraynumspace\\
 \label{EqNovember8at14h22in2024SophiaAntipolisB} 
& = & \iiint  \left(  
\mathsf{L} \left(\vect{u}, \vect{\theta}\right) - 
\mathsf{L}\left(\vect{z}, \vect{\theta}\right)
 \right)  \mathrm{d}P_{\vect{\Theta} | \vect{Z} = \vect{z}}  \left( \vect{\theta} \right)  \mathrm{d}P_{\vect{Z}} \left( \vect{u} \right) \mathrm{d}P_{\vect{Z}} \left( \vect{z} \right) \squeezeequ \IEEEeqnarraynumspace\\
 \nonumber
& = & \iint \left( \int 
\mathsf{L} \left(\vect{u}, \vect{\theta}\right)    \mathrm{d}P_{\vect{Z}} \left( \vect{u} \right) \right)   \mathrm{d}P_{\vect{\Theta} | \vect{Z} = \vect{z}}  \left( \vect{\theta} \right)\mathrm{d}P_{\vect{Z}} \left( \vect{z} \right) \\
\label{EqNovember8at14h22in2024SophiaAntipolisC}  
& & -  \iint   \mathsf{L}\left(\vect{z}, \vect{\theta} \right)  \mathrm{d}P_{\vect{\Theta} | \vect{Z} = \vect{z}}  \left( \vect{\theta} \right)   \mathrm{d}P_{\vect{Z}} \left( \vect{z} \right),
\squeezeequ \IEEEeqnarraynumspace\\
\label{EqNovember8at14h22in2024SophiaAntipolisD}  
& = & \iint \left( \mathsf{R}_{\vect{\theta}} \left( P_{\vect{Z}}\right) -  \mathsf{L}\left(\vect{z}, \vect{\theta} \right)  \right) \mathrm{d}P_{\vect{\Theta} | \vect{Z} = \vect{z}}  \left( \vect{\theta} \right)\mathrm{d}P_{\vect{Z}} \left( \vect{z} \right) ,
\end{IEEEeqnarray}
where the equality in~\eqref{EqNovember8at14h22in2024SophiaAntipolisB} follows from~\cite[Equation~9]{perlaza2024generalization}; 
the equality in~\eqref{EqNovember8at14h22in2024SophiaAntipolisC} follows by exchanging the order of the integrals \cite[Theorem~$2.6.6$]{ash2000probability}; and
the equality in~\eqref{EqNovember8at14h22in2024SophiaAntipolisD} follows from~~\cite[Equation~8]{perlaza2024generalization}.  

The proof continues by noticing that: 
\begin{IEEEeqnarray}{rcl}
\nonumber
& & \iint  \mathsf{R}_{\vect{\theta}} \left( P_{\vect{Z}}\right)  \mathrm{d}P_{\vect{\Theta} | \vect{Z} = \vect{z}}  \left( \vect{\theta} \right)\mathrm{d}P_{\vect{Z}} \left( \vect{z} \right) 
\squeezeequ \IEEEeqnarraynumspace\\ 
\label{EqNovember8at15h23in2024SophiaAntipolisB}  
& = & \iint  \mathsf{R}_{\vect{\theta}} \left( P_{\vect{Z}}\right)  \frac{\mathrm{d}P_{\vect{\Theta} | \vect{Z} = \vect{z}}}{\mathrm{d} P_{\vect{\Theta}}}  \left( \vect{\theta} \right) \mathrm{d} P_{\vect{\Theta}} \left( \vect{\theta} \right) \mathrm{d}P_{\vect{Z}} \left( \vect{z} \right) 
\squeezeequ \IEEEeqnarraynumspace\\
\label{EqNovember8at15h23in2024SophiaAntipolisC}  
& = & \int   \mathsf{R}_{\vect{\theta}} \left( P_{\vect{Z}}\right)  \left( \int \frac{\mathrm{d}P_{\vect{\Theta} | \vect{Z} = \vect{z}}}{\mathrm{d} P_{\vect{\Theta}}}  \left( \vect{\theta} \right)  \mathrm{d}P_{\vect{Z}} \left( \vect{z} \right) \right)   \mathrm{d} P_{\vect{\Theta}} \left( \vect{\theta} \right) \squeezeequ \IEEEeqnarraynumspace \\
\label{EqNovember8at15h23in2024SophiaAntipolisD}  
& = & \int   \mathsf{R}_{\vect{\theta}} \left( P_{\vect{Z}}\right)      \mathrm{d} P_{\vect{\Theta}} \left( \vect{\theta} \right) \squeezeequ \IEEEeqnarraynumspace \\
\label{EqNovember8at15h23in2024SophiaAntipolisE}  
& = & \iint   \left( \frac{1}{n}\sum_{t =1}^{n} \ell\left( x_t, y_t, \vect{\theta}\right) \right) \mathrm{d} P_{\vect{Z}}\left( \vect{z} \right)  \mathrm{d} P_{\vect{\Theta}} \left( \vect{\theta} \right)\\
\label{EqNovember8at15h23in2024SophiaAntipolisF}  
& = &   \int \frac{1}{n}\sum_{t =1}^{n} \int   \ell\left( x, y, \vect{\theta}\right)  \mathrm{d} P_{Z_t}\left( x,y \right)  \mathrm{d} P_{\vect{\Theta}} \left( \vect{\theta} \right)\\
\label{EqOctober14at11h06in2026BenguerirA}  
& = &   \int \frac{1}{n}\sum_{t =1}^{n} \int   \ell\left( x, y, \vect{\theta}\right)  \frac{\mathrm{d}P_{Z_t}}{\mathrm{d}\bar{P}_{Z}}(x,y) \mathrm{d} \bar{P}_{Z}\left( x,y \right)  \mathrm{d} P_{\vect{\Theta}} \left( \vect{\theta} \right)\\
\label{EqOctober14at11h06in2026BenguerirB}  
& = &   \iint  \left( \ell\left( x, y, \vect{\theta}\right) \right) \left( \frac{1}{n}\sum_{t =1}^{n} \frac{\mathrm{d}P_{Z_t}}{\mathrm{d}\bar{P}_{Z}}(x,y) \right)\mathrm{d} \bar{P}_{Z}\left( x,y \right)  \mathrm{d} P_{\vect{\Theta}} \left( \vect{\theta} \right)\\
\label{EqNovember8at15h23in2024SophiaAntipolisG}  
& = & \iint   \left( \ell\left( x, y, \vect{\theta}\right) \right) \mathrm{d} \bar{P}_{Z}\left( x,y \right)  \mathrm{d} P_{\vect{\Theta}} \left( \vect{\theta} \right),
\end{IEEEeqnarray}
where the equality~\eqref{EqNovember8at15h23in2024SophiaAntipolisB} follows from Assumption~$(a)$  and \cite[Theorem~2]{InriaRR9591};
 the equality in~\eqref{EqNovember8at15h23in2024SophiaAntipolisC}  follows by exchanging the order of the integrals \cite[Theorem~$2.6.6$]{ash2000probability};
 the equality in~\eqref{EqNovember8at15h23in2024SophiaAntipolisD} follows from \cite[Theorem~$10$]{InriaRR9591};
the equality in~\eqref{EqNovember8at15h23in2024SophiaAntipolisE} follows from~\cite[Equations~7 and~8]{perlaza2024generalization}; 
the equality in~\eqref{EqNovember8at15h23in2024SophiaAntipolisF} follows from~\cite[Theorem~$1.6.3$]{ash2000probability};
the equality in~\eqref{EqOctober14at11h06in2026BenguerirA} follows from the fact that for all $t \inCountK{n}$, the measure $P_{Z_t}$ is absolutely continuous with $\bar{P}_{Z}$ and thus, ~\cite[Theorem~$2$]{InriaRR9591} holds;
the equality in~\eqref{EqOctober14at11h06in2026BenguerirB} 
follows by exchanging the order of the integrals \cite[Theorem~$2.6.6$]{ash2000probability}; and 
the equality in~\eqref{EqNovember8at15h23in2024SophiaAntipolisG} follows from the fact that for all $(x,y) \in \supp \bar{P}_{Z}$,
\begin{IEEEeqnarray}{rcl}
\label{EqOctober14at11h43in2026BenguerirA}
1 & = & \frac{\mathrm{d}\bar{P}_{Z}}{\mathrm{d}\bar{P}_{Z}} \left(x,y \right)\\
\label{EqOctober14at11h43in2026BenguerirB}
& = & \frac{1}{n} \sum_{t=1}^{n} \frac{\mathrm{d} P_{Z_t}}{\mathrm{d}\bar{P}_{Z}} \left(x,y \right).
\end{IEEEeqnarray}
The equality in~\eqref{EqOctober14at11h43in2026BenguerirA} follows from \cite[Theorem~3]{InriaRR9591}; and the equality in~\eqref{EqOctober14at11h43in2026BenguerirB} follows from \cite[Theorem~6]{InriaRR9591}.

Finally, note that
\begin{IEEEeqnarray}{rcl}
\nonumber
& &  \iint \mathsf{L}\left(\vect{z}, \vect{\theta} \right) \mathrm{d}P_{\vect{\Theta} | \vect{Z} = \vect{z}}  \left( \vect{\theta} \right)\mathrm{d}P_{\vect{Z}} \left( \vect{z} \right) \\
\label{EqNovember10at13h45in2024NiceB}  
& = &  \iint \mathsf{L}\left(\vect{z}, \vect{\theta} \right)  \frac{\mathrm{d}P_{\vect{\Theta} | \vect{Z} = \vect{z}}}{\mathrm{d} P_{\vect{\Theta}}}  \left( \vect{\theta} \right) \mathrm{d} P_{\vect{\Theta}} \left( \vect{\theta} \right) \mathrm{d}P_{\vect{Z}} \left( \vect{z} \right) 
 \IEEEeqnarraynumspace\\
\label{EqNovember10at13h45in2024NiceC}  
& = &  \iint \mathsf{L}\left(\vect{z}, \vect{\theta} \right)  \frac{\mathrm{d}P_{\vect{\Theta} | \vect{Z} = \vect{z}}}{\mathrm{d} P_{\vect{\Theta}}}  \left( \vect{\theta} \right) \mathrm{d}P_{\vect{Z}} \left( \vect{z} \right)  \mathrm{d} P_{\vect{\Theta}} \left( \vect{\theta} \right)
 \IEEEeqnarraynumspace\\
\label{EqNovember10at13h45in2024NiceD}  
& = &  \iint \mathsf{L}\left(\vect{z}, \vect{\theta} \right)  
\frac{\mathrm{d} P_{\vect{Z} | \vect{\Theta} = \vect{\theta}}}{\mathrm{d}P_{\vect{Z}}}  \left( \vect{z} \right) \mathrm{d}P_{\vect{Z} }\left( \vect{z} \right)  \mathrm{d} P_{\vect{\Theta}} \left( \vect{\theta} \right)
 \IEEEeqnarraynumspace\\
\label{EqNovember10at13h45in2024NiceE}  
& = &  \iint \mathsf{L}\left(\vect{z}, \vect{\theta} \right)  \mathrm{d}P_{\vect{Z} | \vect{\Theta} = \vect{\theta}}\left( \vect{z} \right)  \mathrm{d} P_{\vect{\Theta}} \left( \vect{\theta} \right)
 \IEEEeqnarraynumspace\\
\label{EqNovember10at13h45in2024NiceF}  
& = &  \iint\left( \frac{1}{n}\sum_{t =1}^{n} \ell\left( x_t, y_t, \vect{\theta}\right) \right) \mathrm{d} P_{\vect{Z} | \vect{\Theta} = \vect{\theta}}\left( \vect{z} \right)  \mathrm{d} P_{\vect{\Theta}}\left( \vect{\theta}\right) \IEEEeqnarraynumspace\\
\label{EqNovember10at13h45in2024NiceG}  
& = & \int \frac{1}{n}\sum_{t =1}^{n} \int   \ell\left( x, y,  \vect{\theta} \right)  \mathrm{d} P_{Z_t | \vect{\Theta} = \vect{\theta}}\left( x,y \right)  \mathrm{d} P_{\vect{\Theta}}\left( \vect{\theta}\right)\\
\label{EqOctober14at12h12in2026BenguerirA}  
& = & \int \frac{1}{n}\sum_{t =1}^{n} \int   \ell\left( x, y,  \vect{\theta} \right) \frac{\mathrm{d} P_{Z_t | \vect{\Theta} = \vect{\theta}}}{\mathrm{d} \bar{P}_{Z | \vect{\Theta} = \vect{\theta}}} (x,y) \mathrm{d} \bar{P}_{Z | \vect{\Theta} = \vect{\theta}}\left( x,y \right)  \mathrm{d} P_{\vect{\Theta}}\left( \vect{\theta}\right)\\
\label{EqOctober14at12h12in2026BenguerirB}  
& = & \iint   \ell\left( x, y,  \vect{\theta} \right) \left( \frac{1}{n}\sum_{t =1}^{n}  \frac{\mathrm{d} P_{Z_t | \vect{\Theta} = \vect{\theta}}}{\mathrm{d} \bar{P}_{Z | \vect{\Theta} = \vect{\theta}}} (x,y) \right)\mathrm{d} \bar{P}_{Z | \vect{\Theta} = \vect{\theta}}\left( x,y \right)  \mathrm{d} P_{\vect{\Theta}}\left( \vect{\theta}\right)\\
\label{EqNovember10at13h45in2024NiceH}  
& = & \iint   \ell\left( x, y,  \vect{\theta} \right)  \mathrm{d} \bar{P}_{Z | \vect{\Theta} = \vect{\theta}}\left( x,y \right)  \mathrm{d} P_{\vect{\Theta}}\left( \vect{\theta}\right),
\end{IEEEeqnarray}
where 
the equality~\eqref{EqNovember10at13h45in2024NiceB} follows from Assumption~$(a)$ and \cite[Theorem~2]{InriaRR9591}; 
the equality in~\eqref{EqNovember10at13h45in2024NiceC}  follows by exchanging the order of the integrals \cite[Theorem~$2.6.6$]{ash2000probability}; 
the equality~\eqref{EqNovember10at13h45in2024NiceD} follows from Assumption~$(b)$ and \cite[Theorem~11]{InriaRR9591}; 
the equality in~\eqref{EqNovember10at13h45in2024NiceE} follows from~\cite[Theorem~2]{InriaRR9591}; 
the equality in~\eqref{EqNovember10at13h45in2024NiceF} follows from~\cite[Equation~7]{perlaza2024generalization};
the equality in~\eqref{EqNovember10at13h45in2024NiceG} follows from~\cite[Theorem~$1.6.3$]{ash2000probability};
the equality in~\eqref{EqOctober14at12h12in2026BenguerirA} follows from the fact that for all $t \inCountK{n}$, the measure $P_{Z_t | \vect{\Theta} = \vect{\theta}}$ is absolutely continuous with $\bar{P}_{Z| \vect{\Theta} = \vect{\theta}}$ and thus, ~\cite[Theorem~$2$]{InriaRR9591} holds;
the equality in~\eqref{EqOctober14at12h12in2026BenguerirB} 
follows by exchanging the order of the integrals \cite[Theorem~$2.6.6$]{ash2000probability}; and 
the equality in~\eqref{EqNovember10at13h45in2024NiceH} follows from the fact that for all $\vect{\theta} \in \set{M}$ and for all $(x,y) \in \supp \bar{P}_{Z| \vect{\Theta} = \vect{\theta}}$,
\begin{IEEEeqnarray}{rcl}
\label{EqOctober14at12h25in2026BenguerirA}
1 & = & \frac{\mathrm{d}\bar{P}_{Z| \vect{\Theta} = \vect{\theta}}}{\mathrm{d}\bar{P}_{Z| \vect{\Theta} = \vect{\theta}}} \left(x,y \right)\\
\label{EqOctober14at12h25in2026BenguerirB}
& = & \frac{1}{n} \sum_{t=1}^{n} \frac{\mathrm{d} P_{Z_t| \vect{\Theta} = \vect{\theta}}}{\mathrm{d}\bar{P}_{Z| \vect{\Theta} = \vect{\theta}}} \left(x,y \right).
\end{IEEEeqnarray}
The equality in~\eqref{EqOctober14at12h25in2026BenguerirA} follows from \cite[Theorem~3]{InriaRR9591}; and the equality in~\eqref{EqOctober14at12h25in2026BenguerirB} follows from \cite[Theorem~6]{InriaRR9591}.

The proof ends by using~\eqref{EqNovember8at15h23in2024SophiaAntipolisG} and~\eqref{EqNovember10at13h45in2024NiceH} in~\eqref{EqNovember8at14h22in2024SophiaAntipolisD}, which yields
\begin{IEEEeqnarray}{rcl}
\nonumber
 \overline{\overline{\mathsf{G}}} \left(P_{\vect{\Theta} | \vect{Z}}, P_{\vect{Z}} \right)  
& = & \iint   \ell\left( x, y, \vect{\theta}\right)  \mathrm{d} \bar{P}_{Z}\left( x,y \right) \mathrm{d} P_{\vect{\Theta}}\left( \vect{\theta}\right) \\
& & -  \iint  \ell\left( x, y,  \vect{\theta} \right)  \mathrm{d} \bar{P}_{Z | \vect{\Theta} = \vect{\theta}}\left( x,y \right) \mathrm{d} P_{\vect{\Theta}}\left( \vect{\theta}\right) \squeezeequ\spnum\\
\label{EqNovember10at16h19in2024Nice}
& = & \int \mathsf{G} \left( \theta,  \bar{P}_{Z} , \bar{P}_{Z | \vect{\Theta} = \vect{\theta}} \right) \mathrm{d} P_{\vect{\Theta}}\left( \vect{\theta}\right),
\end{IEEEeqnarray}
where the equality in~\eqref{EqNovember10at16h19in2024Nice} follows from~\cite[Equation~21]{perlaza2024generalization}.
This completes the proof.
 
 \section{Proof of Theorem~\ref{TheoremOctober16at19h15in2026atBenguerir}}
\label{AppendixProofOfTheoremOctober16at19h15in2026atBenguerir}

The proof follows from~\cite[Lemma~19]{perlaza2024generalization}, which holds under Assumptions $(a)$, $(b)$, and $(c)$ and does not rely on \cite[Lemma~5]{perlaza2024generalization}.
More specifically, the proof consists of re-writing the difference~$\mathsf{R}_{\vect{\nu}}\left(P_{\vect{Z}|\vect{\Theta}=\vect{\mu}}\right) - \mathsf{R}_{\vect{\mu}}\left(P_{\vect{Z}|\vect{\Theta}=\vect{\mu}}\right)$ in~\cite[Equation~137]{perlaza2024generalization} using the functional~$\mathsf{G}$ in~\cite[Equation~21]{perlaza2024generalization}.
Note that for all $\left( \vect{\nu}, \vect{\mu} \right) \in \set{M} \times \set{M}$, 
	\begin{IEEEeqnarray}{rCl}
\label{EqOctober16at19h12in2026BenguerirA}
\mathsf{R}_{\vect{\nu}}\left(P_{\vect{Z}|\vect{\Theta}=\vect{\mu}}\right) &=& \int \mathsf{L}\left( \vect{z}, \vect{\nu} \right) \mathrm{d} P_{\vect{Z}|\vect{\Theta}=\vect{\mu}} \left(\vect{z} \right) \\
\label{EqOctober16at19h12in2026BenguerirB}
	&=& \int \frac{1}{n} \sum_{t =1 }^{n}\ell\left(\vect{\nu}, x_t,y_t \right) \mathrm{d}P_{\vect{Z}|\vect{\Theta}=\vect{\mu}}( \vect{z})\\
	\label{EqOctober16at19h12in2026BenguerirC}
	&=&  \frac{1}{n} \sum_{t =1 }^{n} \int \ell\left(\vect{\nu}, x,y \right) \mathrm{d}P_{Z_t|\vect{\Theta}=\vect{\mu}}( x,y)\\
	\label{EqOctober16at19h12in2026BenguerirD}
	&=&  \frac{1}{n} \sum_{t =1 }^{n} \int \ell\left(\vect{\nu}, x,y \right) \frac{\mathrm{d}P_{Z_t|\vect{\Theta}=\vect{\mu}}}{\mathrm{d}\bar{P}_{Z|\vect{\Theta}=\vect{\mu}}}( x,y)\mathrm{d}\bar{P}_{Z|\vect{\Theta}=\vect{\mu}}( x,y)\\
	\label{EqOctober16at19h12in2026BenguerirE}
&=&   \int \ell\left(\vect{\nu}, x,y \right) \left( \frac{1}{n} \sum_{t =1 }^{n} \frac{\mathrm{d}P_{Z_t|\vect{\Theta}=\vect{\mu}}} {\mathrm{d}\bar{P}_{Z|\vect{\Theta}=\vect{\mu}}}( x,y) \right)\mathrm{d}\bar{P}_{Z|\vect{\Theta}=\vect{\mu}}( x,y)\\
\label{EqOctober16at19h12in2026BenguerirF}
&=&   \int \ell\left(\vect{\nu}, x,y \right)\mathrm{d}\bar{P}_{Z|\vect{\Theta}=\vect{\mu}}( x,y), 
		\end{IEEEeqnarray}
where 
the equality in~\eqref{EqOctober16at19h12in2026BenguerirA} follows from~\cite[Equation~8]{perlaza2024generalization};
the equality in~\eqref{EqOctober16at19h12in2026BenguerirB} follows from~\cite[Equation~7]{perlaza2024generalization};
the equality in~\eqref{EqOctober16at19h12in2026BenguerirC} follows from~\cite[Theorem~$1.6.3$]{ash2000probability};
the equality in~\eqref{EqOctober16at19h12in2026BenguerirD} follows from~\cite[Theorem~2]{InriaRR9591}; and 
the equality in~\eqref{EqOctober16at19h12in2026BenguerirE} follows from~\eqref{EqOctober14at12h25in2026BenguerirB}.

Hence, the following holds for all~$\left( \vect{\mu} , \vect{\nu} \right) \in \set{M}\times \set{M}$,
	\begin{IEEEeqnarray}{rCl}
		\nonumber
	&&\mathsf{R}_{\vect{\nu}}\left(P_{\vect{Z}|\vect{\Theta}=\vect{\mu}}\right)\\
\label{EqNovember3at20h51in2024Nice}
	&=& 	\mathsf{G}\left(\vect{\nu},\bar{P}_{Z|\vect{\Theta}=\vect{\mu}},P^{(P_{S},\beta)}_{\hat{Z}|\vect{\Theta}=\vect{\nu}} \right)  + \int \ell\left(\vect{\nu},x,y\right)\mathrm{d} P_{\hat{Z} | \vect{\Theta} = \vect{\nu}}^{\left( P_{S}, \beta \right)} (x,y),  \supersqueezeequ \spnum
	\end{IEEEeqnarray}	
	and 
	\begin{IEEEeqnarray}{rCl}
		\nonumber
	&&\mathsf{R}_{\vect{\mu}} \left(P_{\vect{Z}|\vect{\Theta}=\vect{\mu}}\right)\\
	\label{EqNovember3at21h00in2024Nice}
	&=& \mathsf{G}\left(\vect{\mu},\bar{P}_{Z|\vect{\Theta}=\vect{\mu}},P^{(P_{S},\beta)}_{\hat{Z}|\vect{\Theta}=\vect{\mu}} \right)  + \int \ell\left(\vect{\mu},x,y\right)\mathrm{d} P_{\hat{Z} | \vect{\Theta} = \vect{\mu}}^{\left( P_{S}, \beta \right)} (x,y),  \supersqueezeequ \spnum
	\end{IEEEeqnarray}
where the functional~$\mathsf{G}$ is defined in~\cite[Equation~21]{perlaza2024generalization}.	
From Assumption~$(c)$ and~\cite[Theorem~32]{perlaza2024generalization}, it holds that
\begin{IEEEeqnarray}{rCl}
	\nonumber
&&	\mathsf{G}\left(\vect{\nu},\bar{P}_{Z|\vect{\Theta}=\vect{\mu}},P^{(P_{S},\beta)}_{\hat{Z}|\vect{\Theta}=\vect{\nu}} \right)  \\
\nonumber
	&=& \beta \bigg( \KL{\bar{P}_{Z|\vect{\Theta}=\vect{\mu}}}{P_S} - \KL{\bar{P}_{Z|\vect{\Theta}=\vect{\mu}}}{P^{(P_{S},\beta)}_{\hat{Z}|\vect{\Theta}=\vect{\nu}} } \\
	\label{EqNovember3at21h06in2024Nice}
	&&- \KL{P^{(P_{S},\beta)}_{\hat{Z}|\vect{\Theta}=\vect{\nu}} }{P_S}\bigg),
\end{IEEEeqnarray}
and
\begin{IEEEeqnarray}{rCl}
	\nonumber
		&&\mathsf{G}\left(\vect{\mu},\bar{P}_{Z|\vect{\Theta}=\vect{\mu}},P^{(P_{S},\beta)}_{\hat{Z}|\vect{\Theta}=\vect{\mu}} \right) \\
\nonumber
		&=& \beta \bigg( \KL{\bar{P}_{Z|\vect{\Theta}=\vect{\mu}}}{P_S} - \KL{\bar{P}_{Z|\vect{\Theta}=\vect{\mu}}}{P^{(P_{S},\beta)}_{\hat{Z}|\vect{\Theta}=\vect{\mu}}} \\
	\label{EqNovember3at21h07in2024Nice}
	&&- \KL{P^{(P_{S},\beta)}_{\hat{Z}|\vect{\Theta}=\vect{\mu}}}{P_S}\bigg).
	\end{IEEEeqnarray}

Using these equalities, it follows from~\eqref{EqNovember3at20h51in2024Nice} and~\eqref{EqNovember3at21h00in2024Nice}	that
\begin{IEEEeqnarray}{rCl}
	\nonumber
&&	\mathsf{R}_{\vect{\nu}}\left(P_{\vect{Z}|\vect{\Theta}=\vect{\mu}}\right) - \mathsf{R}_{\vect{\mu}} \left(P_{\vect{Z}|\vect{\Theta}=\vect{\mu}}\right)\\
\nonumber
&=& 	\mathsf{G}\left(\vect{\nu},\bar{P}_{Z|\vect{\Theta}=\vect{\mu}},P^{(P_{S},\beta)}_{\hat{Z}|\vect{\Theta}=\vect{\nu}} \right)  + \int \ell\left(\vect{\nu},x,y\right)\mathrm{d} P_{\hat{Z} | \vect{\Theta} = \vect{\nu}}^{\left( P_{S}, \beta \right)} (x,y) \supersqueezeequ \spnum\\
\label{EqNovember3at21h05in2024Nice}		
& & - \mathsf{G}\left(\vect{\mu},\bar{P}_{Z|\vect{\Theta}=\vect{\mu}},P^{(P_{S},\beta)}_{\hat{Z}|\vect{\Theta}=\vect{\mu}} \right) - \int \ell\left(\vect{\mu},x,y\right)\mathrm{d} P_{\hat{Z} | \vect{\Theta} = \vect{\mu}}^{\left( P_{S}, \beta \right)} (x,y) \supersqueezeequ \spnum\\
\nonumber
	&=& \beta \bigg( - \KL{\bar{P}_{Z|\vect{\Theta}=\vect{\mu}}}{P^{(P_{S},\beta)}_{\hat{Z}|\vect{\Theta}=\vect{\nu}} } - \KL{P^{(P_{S},\beta)}_{\hat{Z}|\vect{\Theta}=\vect{\nu}} }{P_S}\bigg)\\
\nonumber
	&& -\beta \bigg(-\KL{\bar{P}_{Z|\vect{\Theta}=\vect{\mu}}}{P^{(P_{S},\beta)}_{\hat{Z}|\vect{\Theta}=\vect{\mu}}} -\KL{P^{(P_{S},\beta)}_{\hat{Z}|\vect{\Theta}=\vect{\mu}}}{P_S}\bigg)\supersqueezeequ  \spnum \\
\nonumber
	&&+ \int \ell\left(\vect{\nu},x,y\right)\mathrm{d} P_{\hat{Z} | \vect{\Theta} = \vect{\nu}}^{\left( P_{S}, \beta \right)} (x,y) \\
\label{EqNovember17at13h09in2024Nice}
	&&- \int \ell\left(\vect{\mu},x,y\right)\mathrm{d} P_{\hat{Z} | \vect{\Theta} = \vect{\mu}}^{\left( P_{S}, \beta \right)} (x,y) \Dsupersqueezeequ  \spnum\\
	\nonumber
	&=& \beta \KL{\bar{P}_{Z|\vect{\Theta}=\vect{\mu}}}{P^{(P_{S},\beta)}_{\hat{Z}|\vect{\Theta}=\vect{\mu}}} - \beta  \KL{\bar{P}_{Z|\vect{\Theta}=\vect{\mu}}}{P^{(P_{S},\beta)}_{\hat{Z}|\vect{\Theta}=\vect{\nu}} } \Dsupersqueezeequ  \spnum\\
\nonumber
	&&+ \int \ell\left(\vect{\nu},x,y\right)\mathrm{d} P_{\hat{Z} | \vect{\Theta} = \vect{\nu}}^{\left( P_{S}, \beta \right)} (x,y) - \beta \KL{P^{(P_{S},\beta)}_{\hat{Z}|\vect{\Theta}=\vect{\nu}} }{P_S}\\
\label{EqNovember3at21h05in2024NiceB}		
	&&- \int \ell\left(\vect{\mu},x,y\right)\mathrm{d} P_{\hat{Z} | \vect{\Theta} = \vect{\mu}}^{\left( P_{S}, \beta \right)} (x,y) + \beta \KL{P^{(P_{S},\beta)}_{\hat{Z}|\vect{\Theta}=\vect{\mu}}}{P_S}
\end{IEEEeqnarray}
where the equality in~\eqref{EqNovember17at13h09in2024Nice} follows from both~\eqref{EqNovember3at21h06in2024Nice} and~\eqref{EqNovember3at21h07in2024Nice}.

Finally, using~\cite[Lemma~19]{perlaza2024generalization} and \eqref{EqNovember3at21h05in2024NiceB}, yields
\begin{IEEEeqnarray}{rCl}
	\nonumber
	&&\overline{\overline{\mathsf{G}}} \left(P_{\vect{\Theta} | \vect{Z}}, P_{\vect{Z}} \right)\\
	\nonumber
	&=& \beta \iint \bigg(\KL{\bar{P}_{Z|\vect{\Theta}=\vect{\mu}}}{P^{(P_{S},\beta)}_{\hat{Z}|\vect{\Theta}=\vect{\mu}}} \\
	&&- \KL{\bar{P}_{Z|\vect{\Theta}=\vect{\mu}}}{P^{(P_{S},\beta)}_{\hat{Z}|\vect{\Theta}=\vect{\nu}} }\bigg)\mathrm{d}P_{\vect{\Theta}}(\vect{\nu})\mathrm{d}P_{\vect{\Theta}}(\vect{\mu}),
\end{IEEEeqnarray}
which completes the proof.

\end{document}